\documentclass[11pt]{article}
\usepackage[english]{babel} %
\usepackage[T1]{fontenc}
\usepackage{graphicx,xspace,amsmath,marvosym,enumerate}
\usepackage{wrapfig,enumerate,dsfont,hyperref,subfig,multirow}
\usepackage[table]{xcolor}
\usepackage{booktabs}

\usepackage[boldsans]{concmath} %
\usepackage{euler} %

\usepackage{amsfonts,amssymb,bm,theorem,fullpage}

\newtheorem{lemma}[equation]{Lemma}
\newtheorem{corollary}[equation]{Corollary}

\newtheorem{theorem}[equation]{Theorem}

\newcommand{\qed}{\hfill\ensuremath{\Box}}
\newenvironment{proof}{\noindent\textbf{Proof.}
}{\par\medskip}

\newcommand{\R}{\ensuremath{\mathds R}}
\newcommand{\T}{\ensuremath{\mathcal T}}
\newcommand{\F}{\ensuremath{\mathcal F}}
\newcommand{\figscale}{1.1}

\graphicspath{{figures/}}

\begin{document}

\title{On Universal Point Sets for Planar Graphs}

\author{Jean Cardinal$^{1,}$\thanks{Partially supported by the ESF EUROCORES
    programme EuroGIGA, CRP ComPoSe. The research was carried out while the
    first author was at ETH Z\"urich.} \and Michael
  Hoffmann$^{2,}$\thanks{Partially supported by the ESF EUROCORES programme
    EuroGIGA, CRP GraDR and the Swiss National Science Foundation, SNF Project
    20GG21-134306.} \and Vincent Kusters$^{2,\dag}$}

\date{\small $^1$D\'epartement d'Informatique, Universit\'e Libre de
  Bruxelles, \texttt{jcardin@ulb.ac.be}\\
  $^2$Institute of Theoretical Computer Science, ETH Z\"urich,
  \texttt{\{hoffmann,kustersv\}@inf.ethz.ch}}


\maketitle%
\begin{abstract}
  A set $P$ of points in $\R^2$ is $n$-universal, if every planar graph on $n$
  vertices admits a plane straight-line embedding on $P$. Answering a question
  by Kobourov, we show that there is no $n$-universal point set of size $n$, for
  any $n\ge 15$. Conversely, we use a computer program to show that there exist
  universal point sets for all $n\le 10$ and to enumerate all corresponding
  order types. Finally, we describe a collection $\mathcal{G}$ of $7'393$ planar
  graphs on $35$ vertices that do not admit a simultaneous geometric embedding
  without mapping, that is, no set of $35$ points in the plane supports a plane
  straight-line embedding of all graphs in $\mathcal{G}$.
\end{abstract}
\sloppy

\section{Introduction}
\label{sec:introduction}

We consider plane, straight-line embeddings of graphs in $\R^2$. In those
embeddings, vertices are represented by pairwise distinct points, every edge is
represented by a line segment connecting its endpoints, and no two edges
intersect except at a common endpoint.

An {\em $n$-universal} (or short \emph{universal}) point set for planar graphs
admits a plane straight-line embedding of all graphs on $n$ vertices. A
longstanding open problem is to give precise bounds on the minimum number of
points in an $n$-universal point set. The currently known asymptotic bounds are
apart by a linear factor. On the one hand, it is known that every planar graph
can be embedded on a grid of size $n-1\times
n-1$~\cite{fpp-hdpgg-90,s-epgg-90}. On the other hand, it was shown that at
least $1.235n$ points are necessary~\cite{Kurowski04}, improving earlier bounds
of $1.206n$~\cite{ck-lbsuspg-89} and $n+\sqrt{n}$~\cite{fpp-hdpgg-90}.

The following, somewhat simpler question was asked ten years ago by
Kobourov~\cite{dmo-topp-}: what is the largest value of $n$ for which a
universal point set of size $n$ exists? We prove the following.
\begin{theorem}
\label{thm:main}
There is no $n$-universal point set of size $n$, for any $n\ge 15$.
\end{theorem}

At some point, the Open Problem Project page dedicated to the
problem~\cite{dmo-topp-} mentioned that Kobourov proved there exist 14-universal
point sets of size 14. If this is correct, our bound is tight, and the answer to
the above question is $n=14$. After verification, however, this claim appears to
be unsubstantiated~\cite{Kobourov}. We managed to check that there exist
universal point sets only up to $n\leq 10$. Further investigations are ongoing.

\paragraph{Overview.}
Section~\ref{sec:uni} is devoted to the proof of Theorem~\ref{thm:main}.  It
combines a labeled counting scheme for \emph{planar 3-trees} (also known as
\emph{stacked triangulations}) that is very similar to the one used by Kurowski
in his asymptotic lower bound argument~\cite{Kurowski04} with known lower bounds
on the rectilinear crossing number~\cite{AF05,lvww-cqk-04}.  Note that although
planar 3-trees seem to be useful for lower bounds, a recent preprint from Fulek
and T\'oth~\cite{FT12} shows that there exist $n$-universal point sets of size
$O(n^{5/3})$ for planar 3-trees.

For a collection $\mathcal G=\{G_1,\ldots,G_k\}$ of planar graphs on $n$
vertices, a \emph{simultaneous geometric embedding without mapping} for
$\mathcal G$ is a collection of plane straight-line embeddings $\phi_i:G_i\to
P$ onto the same set $P\subset\R^2$ of $n$ points.

In Section~\ref{sec:sim}, we consider the following problem: what is the largest
natural number $\sigma$ such that every collection of $\sigma$ planar graphs on
the same number of vertices admit a simultaneous geometric embedding without
mapping? From the F\'ary-Wagner Theorem~\cite{f-slrpg-48,w-bzv-36} we know that
$\sigma\ge 1$.  We prove the following upper bound:
\begin{theorem}\label{thm:sim}
  There is a collection of $7'393$ planar graphs on $35$ vertices that do not
  admit a simultaneous plane straight-line embedding without mapping, hence
  $\sigma < 7'393$.
\end{theorem}
To our knowledge these are the best bounds currently known. It is a very
interesting and probably challenging open problem to determine the exact value
of~$\sigma$.

Finally, in Section~\ref{sec:small}, we use a computer program to show that
there exist $n$-universal point sets of size $n$ for all $n\le 10$ and give the
total number of such point sets for each $n$. As a side remark, note that it is
not clear that the property ``there exists an $n$-universal point set of size
$n$'' is monotone in $n$.

\section{Large Universal Point Sets}
\label{sec:uni}

A planar 3-tree is a maximal planar graph obtained by iteratively splitting a
facial triangle into three new triangles with a degree-three vertex, starting
from a single triangle. Since a planar 3-tree is a maximal planar graph, it has
$n$ vertices and $2n-4$ triangular faces and its combinatorial embedding is
fixed up to the choice of the outer face. 

For every integer $n\ge 4$, we define a family $\T_n$ of labeled planar 3-trees
on the set of vertices $[n]:=\{1,\dots,n\}$ as follows:
\begin{enumerate}[(i)]
\item $\T_4$ contains only the complete graph $K_4$,
\item $\T_n$ contains every graph that can be constructed by making the new
vertex $n$ adjacent to the three vertices of one of the $2n-6$ facial triangles
of some $T\in\T_{n-1}$.
\end{enumerate}
We insist on the fact that $\T_n$ is a set of {\em labeled} abstract graphs,
many of which can in fact be isomorphic if considered as abstract (unlabeled)
graphs. We also point out that for $n>4$, the class $\T_n$ does not contain
\emph{all} labeled planar 3-trees on $n$ vertices. For instance, the four graphs
in $\T_5$ are shown in \figurename~\ref{fig:T5}, and there is no graph for which
both Vertex~$1$ and Vertex~$2$ have degree three.

\begin{lemma}\label{lem:tn_size_lower}
  For $n\geq 4$, we have $|\T_n| = 2^{n-4}\cdot
  (n-3)!$.
\end{lemma}
\begin{proof}
  By definition, $|\T_4|=1$. Every graph in $\T_n$ is constructed by splitting
  one of the $2n-6$ faces of a graph in $\T_{n-1}$. We therefore have:
  \[
  \hspace*{3cm}|\T_n|%
  = |\T_{n-1}|\cdot (2n-6)%
  = 4 \cdot 6 \cdot \ldots \cdot (2n - 6)%
  = 2^{n-4}\cdot (n-3)!.\hspace*{3cm}\qed
  \]
\end{proof}

\begin{figure}[htbp]
  \centering%
  \includegraphics[scale=\figscale]{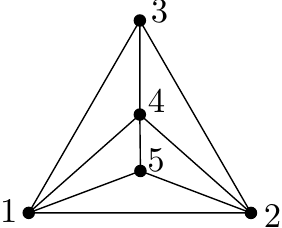}\hfill%
  \includegraphics[scale=\figscale]{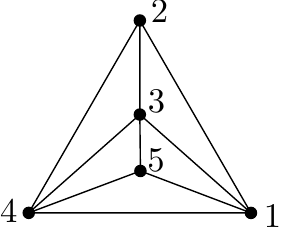}\hfill%
  \includegraphics[scale=\figscale]{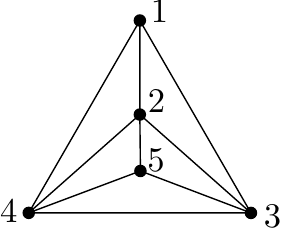}\hfill%
  \includegraphics[scale=\figscale]{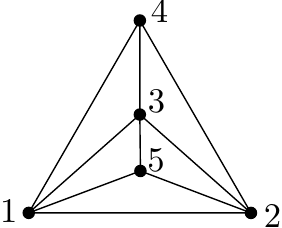}
  \caption{\label{fig:T5}The four planar 3-trees in $\T_5$, with vertex set $\{1,2,3,4,5\}$.}
\end{figure}

\begin{lemma}
  \label{lem:labeled_points_unique_stacked}
  Given a set $P=\{p_1,\dots,p_n\}$ of labeled points in the plane and a
  bijection $\pi: [n] \to P$, there is \emph{at most one} $T\in \T_n$ such that
  $\pi$ is a plane straight-line embedding of $T$.
\end{lemma}
\begin{proof}
  Consider any such labeled point set $P$ and assume without loss of
  generality that $\pi(i)=p_i$ for all $i$. In all $T\in\T_n$ the
  vertices $\{1,2,3,4\}$ form a $K_4$. Hence, for all $T$, the
  straight-line embedding $\pi$ connects all pairs of points in
  $\{p_1,p_2,p_3,p_4\}$ with line segments. If these points are in
  convex position, there is a crossing and there is no $T\in\T_n$ for
  which $\pi$ is a plane straight-line embedding
  (\figurename~\ref{fig:crossingornot}). Otherwise, there is a unique
  graph $K_4\in\T_4$ for which $p_1,p_2,p_3,p_4$ is a plane
  straight-line drawing. We proceed as follows.

  \begin{figure}[htbp]
    \hfil%
    \includegraphics[scale=\figscale]{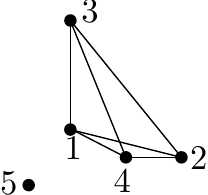}\hfil%
    \includegraphics[scale=\figscale]{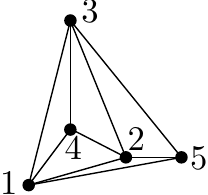}\hfil%
    \caption{\label{fig:crossingornot}Some permutations of a given point set do
      not define any planar 3-tree in $\T_n$, because they generate a crossing
      (left). On the other hand, when no such crossing occurs, the permutation
      defines a unique planar 3-tree in $\T_n$ (right). At any rate, a single
      permutation can be associated with at most one planar 3-tree in $\T_n$.}
  \end{figure}

  Given a plane straight-line drawing for some graph $T_i\in\T_i$ on the first
  $i\ge 4$ points, the next point $p_{i+1}$ is located in some triangular
  region of the drawing; denote this region by $p_{a_i}p_{b_i}p_{c_i}$. Only if
  during the construction of $T$ we decided to connect the next vertex $i+1$ to
  exactly the vertices $a_i,b_i,c_i$, there is no crossing introduced by
  mapping $i+1$ to $p_{i+1}$. (An edge to any other vertex would cross one of
  the bounding edges of the triangle $p_{a_i}p_{b_i}p_{c_i}$.) In other words,
  for every $i\ge 5$ the role of vertex $i$ is completely determined. If no
  crossing is ever introduced, this process determines exactly one graph
  $T\in\T_n$ for which $\pi$ forms a plane straight-line embedding. (Note that
  a crossing can be introduced only if $p_{i+1}$ is located outside of the
  convex hull of $\{p_1,\ldots,p_i\}$. And also in that case there need not be
  a crossing, as the example in \figurename~\ref{fig:crossingornot} (right)
  shows.)\qed
\end{proof}

\noindent We use the following theorem by \'{A}brego and Fern\'{a}ndez-Merchant.
\begin{theorem}[\cite{AF05}]\label{t:cn}
  Every plane straight-line drawing of the complete graph $K_n$ has at
  least $\frac 1 4%
  \left\lfloor\frac{n }{2}\right\rfloor%
  \left\lfloor\frac{n-1}{2}\right\rfloor%
  \left\lfloor\frac{n-2}{2}\right\rfloor%
  \left\lfloor\frac{n-3}{2}\right\rfloor$ crossings.
\end{theorem}
Note that for $n\le 4$ at least one of the floor expressions is zero, whereas
for $n=5$ the theorem states that every straight-line drawing of $K_5$ has at
least one crossing. Any pair of crossing edges corresponds to a four-tuple of
points in convex position. Using this interpretation we can easily derive a
floor-free lower bound on the number of convex four-gons contained in every
planar point set.
\begin{corollary}\label{cor:convex_4_tuples}
  Given a point set $P\subset\R^2$ of $n$ points in general position,
  more than a $\frac{3}{8}\cdot\frac{n-4}{n}$-fraction of all four
  element subsets of $P$ is in convex position.
\end{corollary}
\begin{proof}
  By Theorem~\ref{t:cn} at least %
  $c = \frac 1 4%
  \left\lfloor\frac{n }{2}\right\rfloor%
  \left\lfloor\frac{n-1}{2}\right\rfloor%
  \left\lfloor\frac{n-2}{2}\right\rfloor%
  \left\lfloor\frac{n-3}{2}\right\rfloor$ %
  four element subsets of $P$ are in convex position. For $n$ odd we
  have
  \[c=\frac{1}{4}\left(\frac{n-1}{2}\right)^2\left(\frac{n-3}{2}\right)^2\]
  and for $n$ even we have
  \[c=\frac{1}{4}\left(\frac{n}{2}\right)\left(\frac{n-2}{2}\right)^2\left(\frac{n-4}{2}\right)\]
  and so
  \begin{eqnarray*}
    c &>&
    \frac{1}{4}\left(\frac{n-1}{2}\right)\left(\frac{n-2}{2}\right)\left(\frac{n-3}{2}\right)\left(\frac{n-4}{2}\right)\\
    &=&
    \frac{3}{8}\cdot\frac{n-4}{n}\cdot\frac{n(n-1)(n-2)(n-3)}{4\cdot
      3\cdot 2}\\
    &=&
    \frac{3}{8}\cdot\frac{n-4}{n}\cdot\binom{n}{4}\,,
  \end{eqnarray*}
  for all $n$.\qed
\end{proof}
We will use this fact to prove the following lemma.
\begin{lemma}\label{lem:tn_on_point_set}
  On any set $P\subset\R^2$ of $n\ge 4$ points fewer than
  $\frac{1}{8}(5n+12)(n-1)!$ graphs from $\T_n$ admit a plane
  straight-line embedding.
\end{lemma}
\begin{proof}
  Let $P\subset\R^2$ be a set of $n$ points and denote by
  $\F_n\subseteq\T_n$ the set of labeled planar $3$-trees from $\T_n$
  that admit a plane straight-line embedding onto $P$. Note that a
  straight-line embedding can be represented by a permutation $\pi$ of
  the points of $P$, where each vertex $i$ is mapped to point
  $\pi(i)$. Let $S_n$ be the set of all permutations of $P$. We define
  a map $\psi : \F_n \to S_n$ by assigning to each $T\in\F_n$ some
  $\psi(T)\in S_n$ such that $\psi(T)$ is a plane straight-line
  embedding of $T$ (such an embedding exists by definition of $\F_n$).

  By Lemma~\ref{lem:labeled_points_unique_stacked}, every permutation
  $\pi\in S_n$ is a plane straight-line embedding of \emph{at most
    one} $T\in \F_n$. It follows that $\psi$ is a injection, and hence
  $\psi:\F_n\to\Pi$, with $\Pi=\mathrm{Im}(\psi)$, is a bijection and so  
  $|\F_n|=|\Pi|\le|S_n|=n!$.

  Next we can quantify the difference between $\Pi$ and $S_n$ using
  Corollary~\ref{cor:convex_4_tuples}. Note that the general position
  assumption is not a restriction, since in case of collinearities, a
  slight perturbation of the point set yields a new point set that
  still admits all plane straight-line drawings of the original point
  set. Consider a permutation $\pi=p_1,\ldots,p_n$ such that
  $p_1,p_2,p_3,p_4$ form a convex quadrilateral. As argued in the
  first paragraph of the proof of
  Lemma~\ref{lem:labeled_points_unique_stacked}, $\pi$ is not a plane
  straight-line embedding for any $T\in\F_n$. It follows that $\pi\in
  S_n\setminus\Pi$. We know from Corollary~\ref{cor:convex_4_tuples}
  that more than a fraction of $(3/8)\cdot(n-4)/n$ of the 4-tuples of
  $P$ are in convex position and therefore a corresponding fraction of
  all permutations does \emph{not} correspond to a plane straight-line
  drawing. So we can bound the number of possible labeled plane
  straight-line drawings by
  \[
  \hspace*{4.25cm}|\Pi|<\left(1-\frac{3}{8} \cdot \frac{n-4}{n}\right) n!=
  \frac{1}{8}(5n+12)(n-1)!\,.\hspace*{4.25cm}\qed
  \]
\end{proof}

\begin{proof}[of Theorem~\ref{thm:main}]
  Consider an $n$-universal point set $P\subset\R^2$ with
  $|P|=n$. Being universal, in particular $P$ has to accommodate all
  graphs from $\T_n$. By Lemma~\ref{lem:tn_size_lower}, there are
  exactly $2^{n-4}\cdot (n-3)!$ graphs in $\T_n$, whereas by
  Lemma~\ref{lem:tn_on_point_set} no more than
  $\frac{1}{8}(5n+12)(n-1)!$ graphs from $\T_n$ admit a plane
  straight-line drawing on $P$. Combining both bounds we obtain
  $2^{n-1}\le(5n+12)(n-1)(n-2)$. Setting $n=15$ yields
  $2^{14}=16'384\le 87\cdot 14\cdot 13=15'834$, which is a
  contradiction and so there is no $15$-universal set of $15$ points.

  For $n=14$ the inequality reads $2^{13}=8'192\le 82\cdot 13\cdot
  12=12'792$ and so there is no indication that there cannot be a
  $14$-universal set of $14$ points. To prove the claim for any
  $n>15$, consider the two functions $f(n)=2^{n-1}$ and
  $g(n)=(5n+12)(n-1)(n-2)$ that constitute the inequality. As $f$ is
  exponential in $n$ whereas $g$ is just a cubic polynomial, $f$
  certainly dominates $g$, for sufficiently large $n$.  Moreover, we
  know that $f(15)>g(15)$. Noting that $f(n)/f(n-1)=2$ and $g(n)>0$,
  for $n>2$, it suffices to show that $g(n)/g(n-1)<2$, for all $n\ge
  16$.  

  \noindent We can bound
  \begin{eqnarray*}
    \frac{g(n)}{g(n-1)} &=&
    \frac{(5n+12)(n-1)(n-2)}{(5(n-1)+12)(n-2)(n-3)}=
    \frac{(5n+12)(n-1)}{(5n+7)(n-3)}\\
    &<& \frac{(5n+15)n}{5n(n-3)}=\frac{n+3}{n-3},
  \end{eqnarray*}
  which is easily seen to be upper bounded by two, for $n\geq 9$.\qed
\end{proof}

\section{Simultaneous Geometric Embeddings}
\label{sec:sim}

The number of non-isomorphic planar 3-trees on $n$ vertices was computed by
Beineke and Pippert~\cite{BP74}, and appears as sequence A027610 on Sloane's
Encyclopedia of Integer Sequences. For $n=15$, this number is $321'776$.  Hence
we can also phrase our result in the language of simultaneous
embeddings~\cite{bcdeeiklm-spge-07}.
\begin{corollary}
  There is a collection of $321'776$ planar graphs that do not admit a
  simultaneous (plane straight-line) embedding without mapping.
\end{corollary}
In the following we will give an explicit construction for a much smaller family
of graphs that not admit a simultaneous embedding without mapping. As a first
observation, note that the freedom to select the outer face is essential in
order to embed graphs onto a given point set. In fact, for planar 3-trees, the
mapping for the outer face is the only choice there is. We prove this in two
steps.
\begin{lemma}
  \label{lem:stacked_any_face}
  Let $G$ be a labeled planar 3-tree on the vertex set $[n]$, for $n\ge 3$, and
  let $C$ denote any triangle in $G$. Then $G$ can be constructed starting from
  $C$ by iteratively inserting a degree-three vertex into some facial triangle
  of the partial graph constructed so far.
\end{lemma}
\begin{proof}
  We prove the statement by induction on $n$. For $n=3$ there is nothing to
  show. Hence let $n>3$. By definition $G$ can be constructed iteratively from
  \emph{some} triangle in the way described. Without loss of generality suppose
  that adding vertices in the order $1,2,\ldots,n$ yields such a construction
  sequence. Denote by $G_i$ the graph that is constructed by the sequence
  $1,\ldots,i$, for $1\le i\le n$.

  Let $C=u,v,w$ such that $u<v<w$. Consider the graph $G_w$: In the last step,
  $w$ is added as a new vertex into some facial triangle $T$ of $G_{w-1}$. As
  $w$ is a neighbor of both $u$ and $v$ in $G$, both $u$ and $v$ are vertices of
  $T$; denote the third vertex of $T$ by $x$. Note that all of $u,v,w$ and
  $u,w,x$ and $v,x,w$ are facial triangles in $G_w$.

  If $w=4$, then exchanging the role of $w$ and $x$ yields a construction
  sequence $u,v,w,x,5,\ldots,n$ for $G$, as claimed. If $w>4$, then $c_1,c_2,v$
  is a separating triangle in $G_w$. By the inductive hypothesis we can obtain a
  construction sequence $S$ for $G_{w-1}$ starting with the triangle
  $u,v,x$. The desired sequence for $G$ is obtained as
  $u,v,w,x,S^-,w+1,\ldots,n$, where $S^-$ is the suffix of $S$ that excludes the
  starting triangle $u,v,x$.\qed
\end{proof}
And now we can prove the desired property:
\begin{lemma}
  \label{lem:stacked_unique}
  Given a labeled planar 3-tree $G$ on vertex set $[n]$, a triangle
  $c=c_1c_2c_3$ in $G$, and a set $P\subset \R^2$ of $n$ points with
  $p_1,p_2,p_3\in P$, there is at most one way to complete the partial
  embedding $\{c_1\mapsto p_1, c_2\mapsto p_2, c_3\mapsto p_3\}$ to a plane
  straight-line embedding of $G$ on $P$.
\end{lemma}
\begin{proof}
  We use Lemma~\ref{lem:stacked_any_face} to relabel the vertices in such a way
  that $c_1,c_2,c_3$ becomes $1,2,3$ and the order $1,\dots,n$ is a construction
  sequence for $G$. Embed vertices $1,2,3$ onto $p_1,p_2,p_3$. We iteratively
  embed the remaining vertices as follows. Vertex $i$ was inserted into some
  face $jk\ell$ during the construction given by
  Lemma~\ref{lem:stacked_any_face}. Note that $j,k,\ell$ have already been
  embedded on points $p_j,p_k,p_\ell$. The vertices contained in the triangle
  $jk\ell$ (except $i$) are partitioned into three sets by the cycles $ijk$
  ($n_1$ vertices) and $ik\ell$ ($n_2$ vertices) and $i\ell j$ ($n_3$
  vertices). We want to embed $i$ on a point $p_i$ such that $p_ip_jp_k$
  contains exactly $n_1$ points, $p_ip_kp_\ell$ contains exactly $n_2$ points
  and $p_ip_{\ell}p_j$ contains exactly $n_3$ points. Note that it is necessary
  to embed $i$ on a point with this property: if some triangle has too few
  points, then it will not be possible to embed the subgraph of $G$ enclosed by
  the corresponding cycle there. It remains to show that there is always at most
  one choice for $p_i$. Suppose that there are two candidates for $p_i$, say
  $p_i'$ and $p_i''$. Then $p_i''$ must be contained in $p_i'p_jp_k$ or
  $p_i'p_kp_\ell$ or $p_i'p_{\ell}p_j$ (or vice versa). Without loss of
  generality, let it be contained in $p_i'p_jp_k$: now $p_i''p_jp_k$ contains
  fewer points than $p_i'p_jp_k$, which is a contradiction. The lemma follows by
  induction.\qed
\end{proof}
Therefore it is not surprising that it is very easy to find three graphs that do
not admit a simultaneous (plane straight-line) embedding without mapping, if the
mapping for the outer face is specified for each of them.
\begin{figure}[htbp]
  \hfil%
  \includegraphics[scale=\figscale]{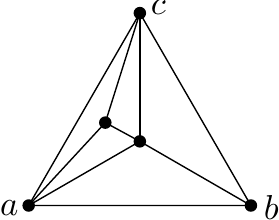}\hfil%
  \includegraphics[scale=\figscale]{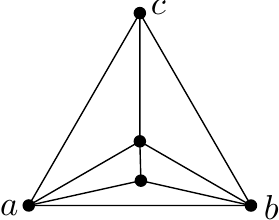}\hfil%
  \includegraphics[scale=\figscale]{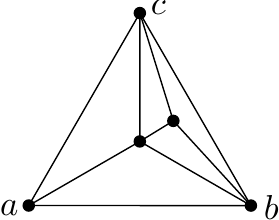}\hfil%
  \caption{\label{fig:three}Three planar graphs that do not admit a simultaneous
    geometric embedding with a fixed mapping for the outer face.}
\end{figure}
\begin{lemma}\label{prop:three}
  There is no set $P\subset\R^2$ of five points with convex hull $p_a,p_b,p_c$
  such that every graph shown in \figurename~\ref{fig:three} has a (plane
  straight-line) embedding on $P$ where the vertices $a$, $b$ and $c$ are mapped
  to the points $p_a$, $p_b$ and $p_c$, respectively.
\end{lemma}
\begin{proof}
  The point $p$ for the central vertex that is connected to all of $a,b,c$ must
  be chosen so that (i) it is not in convex position with $p_a$, $p_b$ and $p_c$
  and (ii) the number of points in the three resulting triangles is one in one
  triangle and zero in the other two. That requires three distinct choices for
  $p$, but there are only two points available.\qed
\end{proof}
In fact, there are many such triples of graphs. The following lemma can be
verified with help of a computer program that exhaustively checks all order
types. Point set order types~\cite{gp-ms-83} are a combinatorial abstraction of
planar point sets that encode the orientation of all point triples, which in
particular determines whether or not any two line segments cross. For a small
number of points, there is a database with realizations of every (realizable)
order type~\cite{ak-psotd-01}.
\begin{lemma}\label{prop:seven}
  There is no set $P\subset\R^2$ of eight points with convex hull $p_a,p_b,p_c$
  such that every graph shown in \figurename~\ref{fig:seven} has a (plane
  straight-line) embedding on $P$ where the vertices $a$, $b$ and $c$ are mapped
  to the points $p_a$, $p_b$ and $p_c$, respectively.
\end{lemma}
\begin{figure}[htbp]
  \hfil%
  \subfloat[$T_1$]{\includegraphics[scale=1.1]{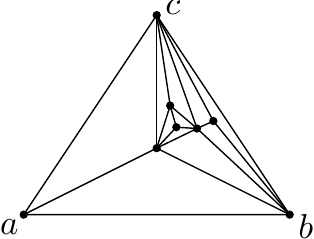}}\hfil%
  \subfloat[$T_2$]{\includegraphics[scale=1.1]{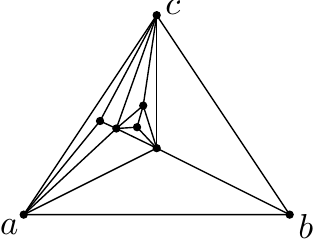}}\hfil%
  \subfloat[$T_3$]{\includegraphics[scale=1.1]{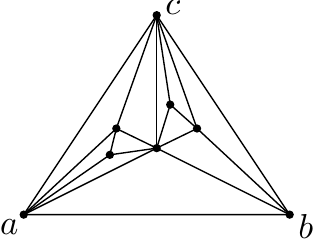}}\hfil%
  \subfloat[$T_4$]{\includegraphics[scale=1.1]{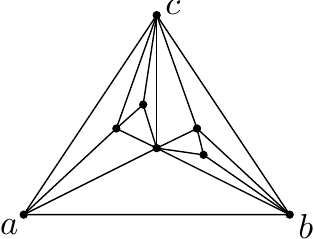}}\hfil\\
  \hspace*{1pt}\hfil%
  \subfloat[$T_5$]{\includegraphics[scale=1.1]{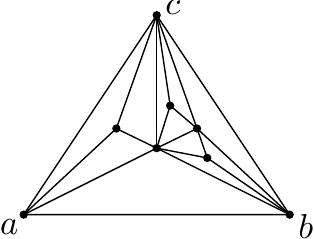}}\hfil%
  \subfloat[$T_6$]{\includegraphics[scale=1.1]{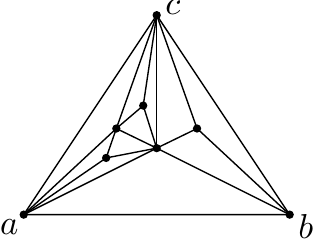}}\hfil%
  \subfloat[$T_7$]{\includegraphics[scale=1.1]{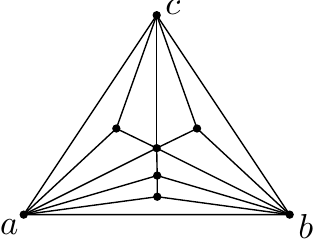}}\hfil%
  \subfloat[$B$]{\includegraphics[scale=1.1]{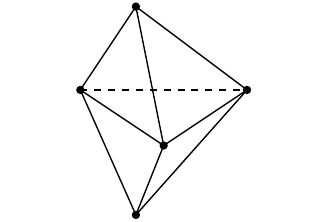}\hfil}\\%
  \caption{\label{fig:seven}{\small (a)--(g)}: Seven planar graphs, no three of
    which admit a simultaneous geometric embedding with a fixed mapping for the
    outer face; {\small (h)}: the skeleton $B$ of a triangular bipyramid.}
\end{figure}
Denote by $\mathcal{T}=\{T_1,\ldots,T_7\}$ the family of seven graphs on eight
vertices depicted in \figurename~\ref{fig:seven}. We consider these graphs as
abstract but \emph{rooted} graphs, that is, one face is designated as the outer
face and the counterclockwise order of the vertices along the outer face (the
\emph{orientation} of the face) is $a,b,c$ in each case. Observe that all graphs
in $\mathcal{T}$ are planar 3-trees.

Using $\mathcal{T}$ we construct a family $\mathcal{G}$ of graphs as follows.
Start from the skeleton $B$ of a triangular bipyramid, that is, a triangle and
two additional vertices, each of which is connected to all vertices of the
triangle. The graph $B$ has five vertices and six faces and it is a planar
3-tree.

We obtain $\mathcal{G}$ from $B$ by planting one of the graphs from
$\mathcal{T}$ onto each of the six faces of $B$. Each face of $B$ is a
(combinatorial) triangle where one vertex has degree three (one of the pyramid
tips) and the other two vertices have degree four (the vertices of the starting
triangle). On each face $f$ of $B$ a selected graph $T$ from $\mathcal{T}$ is
planted by identifying the three vertices bounding $f$ with the three vertices
bounding the outer face of $T$ in such a way that vertex $c$ (which appears at
the top in \figurename~\ref{fig:seven}) is mapped to the vertex of degree three
(in $B$) of $f$. In the next paragraph, we will see why we do not have to
specify how $a$ and $b$ are matched to $f$. The family $\mathcal{G}$ consists of
all graphs on $5+6\cdot 5=35$ vertices that can be obtained in this way. By
construction all these graphs are planar 3-trees. Therefore by
Lemma~\ref{lem:stacked_unique} on any given set of $35$ points, the plane
straight-line embedding is unique (if it exists), once the mapping for the outer
face is determined.

Observe that $\mathcal{T}$ is \emph{flip-symmetric} with respect to horizontal
reflection. In other (more combinatorial) words, for every $T\in\mathcal{T}$ we
can exchange the role of the bottom two vertices $a$ and $b$ of the outer face
(and thereby also its orientation) to obtain a graph that is also in
$\mathcal{T}$. The graphs form symmetric pairs of siblings $(T_1,T_2)$,
$(T_3,T_4)$, $(T_5,T_6)$, and $T_7$ flips to itself. Therefore, regardless of
the orientation in which we plant a graph from $\mathcal{T}$ onto a face of $B$,
we obtain a graph in $\mathcal{G}$, and so $\mathcal{G}$ is well-defined.

Next, we give a lower bound on the number of nonisomorphic graphs in
$\mathcal{G}$.
\begin{lemma}
  The family $\mathcal{G}$ contains at least $9'805$ pairwise nonisomorphic
  graphs.
\end{lemma}
\begin{proof}
  Consider the bipyramid $B$ as a face-labeled object. There are $7^6$ different
  ways to assign a graph from $\mathcal{T}$ to each of the six now
  distinguishable faces. Denote this class of face-labeled graphs by
  $\mathcal{F}$. For many of these assignments the corresponding graphs are
  isomorphic if considered as abstract (unlabeled) graphs. However, the
  following argument shows that every isomorphism between two such graphs maps
  the vertex set of $B$ to itself.

  The two tips of $B$ have degree three and are incident to three faces. Onto
  each of the faces one graph from $\mathcal{T}$ is planted, which increases the
  degree by four (for $T_1,\ldots,T_6$) or three (for $T_7$) to a total of at
  least twelve. The three triangle vertices start with degree four and are
  incident to four faces. Every graph from $\mathcal{T}$ planted there adds at
  least one more edge, to a total degree of at least eight. But the highest
  degree among the interior vertices of the graphs in $\mathcal{T}$ is seven,
  which proves the claim.

  Hence we have to look for isomorphisms only among the symmetries of the
  bipyramid $B$. The tips are distinguishable from the triangle vertices,
  because the former are incident to three high degree vertices, whereas the
  latter are incident to four high degree vertices. Selecting the mapping for
  one face of $B$ determines the whole isomorphism. Since there are at most two
  ways to map a face to a face (we can select the mapping for the two non-tip
  vertices, that is, the orientation of the triangle), every graph in
  $\mathcal{F}$ is isomorphic to at most $2\cdot 6=12$ graphs from
  $\mathcal{F}$. It follows that there are at least $7^6/12>9'804$ pairwise
  nonisomorphic graphs in $\mathcal{G}$.\qed
\end{proof}
We now give an upper bound on the number of graphs of $\mathcal{G}$ that can be
simultaneously embedded on a common point set.
\begin{lemma}
  At most $7'392$ pairwise nonisomorphic graphs of $\mathcal{G}$ admit a
  simultaneous (plane straight-line) embedding without mapping.
\end{lemma}
\begin{proof}
  Consider a subset $\mathcal{G}'\subseteq \mathcal{G}$ of pairwise
  nonisomorphic graphs and a point set $P$ that admits a simultaneous embedding
  of $\mathcal{G}'$. Since $\mathcal{G}'$ is a class of maximal planar graphs,
  the convex hull of $P$ must be a triangle. For each $G\in\mathcal{G}'$ we can
  select an outer face $f(G)$ and a mapping $\pi(G)$ for the vertices bounding
  $f(G)$ to the convex hull of $P$ so that the resulting straight-line
  embedding, which by Lemma~\ref{lem:stacked_unique} is completely determined by
  $f(G)$ and $\pi(G)$, is plane.

  Let us group the graphs from $\mathcal{G}'$ into bins, according to the maps
  $f$ and $\pi$. For $f$, there are $7\cdot 11$ possible choices: one of the
  eleven faces of one of the seven graphs in $\mathcal{T}$. For $\pi$ there are
  three choices: one of the three possible rotations to map the face chosen by
  $f$ to the convex hull of $P$. Note that regarding $\pi$ there is no
  additional factor of two for the orientation of the face, because by
  flip-symmetry such a change corresponds to a different graph (for
  $T_1,\ldots,T_6$) or a different face of the graph (for $T_7$), that is, a
  different choice for $f$. Altogether this yields a partition of $\mathcal{G}'$
  into $3\cdot 77=231$ bins.

  The crucial observation (and ultimate reason for this subdivision) is that for
  all graphs in a single bin the vertices of $B$ (the bipyramid) are mapped to
  the same points. This is a consequence of the uniqueness of the embedding up
  to the mapping for the outer face (Lemma~\ref{lem:stacked_unique}), which is
  identical for all graphs in the same bin. Therefore, the triangle $t$ of $B$
  in which the outer face is located is mapped to the same oriented triple of
  points in $P$ for all graphs in the same bin. From there the pattern repeats,
  noting that every face of $B$ contains the same number of points (five) and
  that the polyhedron $B$ is face-transitive so that there is no difference as
  to which face of $B$ was selected to contain the outer face.

  It follows that for all graphs in the same bin the graphs from $\mathcal{T}$
  planted onto the faces of $B$ are mapped to the same point sets. Any two
  (nonisomorphic) graphs from $\mathcal{G}'$ differ in at least one of those
  faces -- and by definition not in the one in which the outer face was selected
  by $f$. In order for the graphs in a bin to be simultaneously embeddable on
  $P$, by Lemma~\ref{prop:seven} there are at most two different graphs from
  $\mathcal{T}$ mapped to any of the remaining five faces of $B$. Therefore
  there cannot be more than $2^5=32$ graphs from $\mathcal{G}'$ in any bin.
  Hence $|\mathcal{G}'|\leq 231\cdot 32 = 7'392$, as claimed.\qed
\end{proof}
Since there are strictly more nonisomorphic graphs in $\mathcal{G}$ than can
possibly be simultaneously embedded, not all graphs of $\mathcal{G}$ admit a
simultaneous embedding.  In particular, any subset of $7'392 + 1$ nonisomorphic
graphs in $\mathcal{G}$ is a collection that does not have a simultaneous
embedding. This proves our Theorem~\ref{thm:sim}.

\section{Small $n$-universal point sets}
\label{sec:small}

As we have seen in the previous sections, there are no $n$-universal point sets
of size $n$ for $n\geq 15$. In this section, we consider the case $n<15$.
Specifically, we used a computer program~\cite{universalsrc} to show the
following:
\begin{theorem}
  \label{thm:at_most_10}
  There exist $n$-universal point sets of size $n$ for all $1\leq n\leq 10$.
\end{theorem}
We use a straightforward brute-force approach. The two main ingredients are the
aforementioned order type database~\cite{ak-psotd-01} with point sets of size
$n\leq 10$ and the \emph{plantri} program for generating maximal planar
graphs~\cite{brinkmann2007fast,plantrisrc}. To determine if a point set $P$ of
size $n$ is $n$-universal, our program tests if for all maximal planar graphs
$G=(V,E)$ on $n$ vertices, there exists a bijection $\varphi : V\rightarrow P$ such
that straight-line drawing of $G$ induced by $\varphi$ is plane. If such a
bijection exists for all $G$, then $P$ is universal. Otherwise, there is a graph
$G$ that has no plane straight-line embedding on $P$. Note that it is sufficient
to consider maximal planar graphs since adding edges only makes the embedding
problem more difficult. Work on the case $n=11$ is still in progress at the time
of writing. For $n>11$ the approach unfortunately becomes infeasible; it is
unknown whether or not there exist $n$-universal point sets of size $n$ for
$11\leq n\leq 14$. Table~\ref{tab:ups_numbers} gives an overview of the results
of this paper and \figurename~\ref{fig:universal_5_10} shows one universal point
set for each $n=5,\dots,10$.

\begin{table}
  \centering
  \setlength{\tabcolsep}{2.5pt}
  \begin{tabular}{@{}r*{15}{c}@{}}
    \toprule%
    $n$: & $1$ & $2$ & $3$ & $4$ & $5$ & $6$ & $7$ & $8$ & $9$ & $10$ & $11$ & $12$ &
    $13$ & $14$ & $\geq 15$\\
    \# universal point sets: & $1$ & $1$ & $1$ & $1$ & $1$ & $5$ & $45$ & $364$ & $5'955$ & $2'072$ & $?$ &
    $?$ & $?$ & $?$ & $0$\\
    \bottomrule\\
  \end{tabular}
  \caption{The number of (non-equivalent) $n$-universal point sets of size $n$.}
  \label{tab:ups_numbers}
\end{table}

\begin{figure}[htbp]
  \centering%
  \newcommand{\thisfigwidth}{0.3\textwidth}%
  \includegraphics[width=\thisfigwidth]{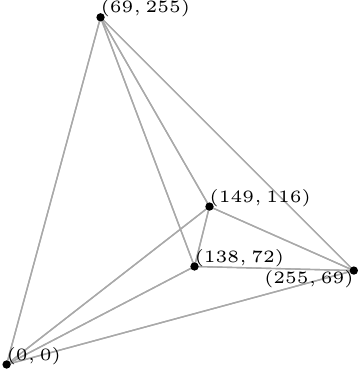}\hfil%
  \includegraphics[width=\thisfigwidth]{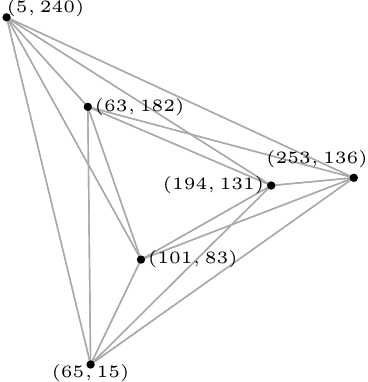}\hfil%
  \includegraphics[width=\thisfigwidth]{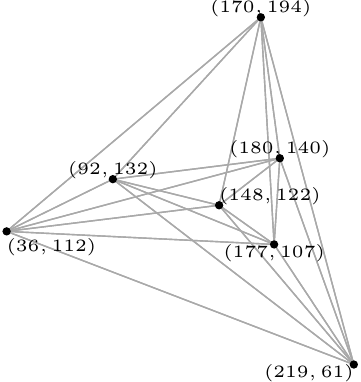}\\[\baselineskip]
  \includegraphics[width=\thisfigwidth]{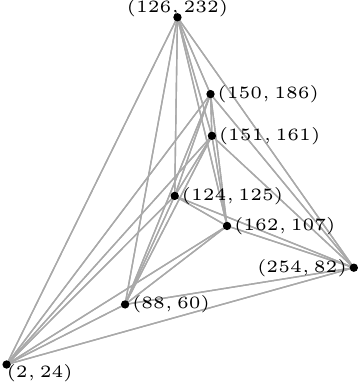}\hfil%
  \includegraphics[width=\thisfigwidth]{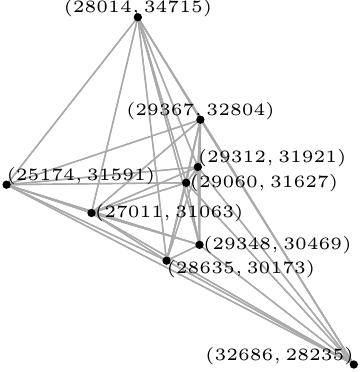}\hfil%
  \includegraphics[width=\thisfigwidth]{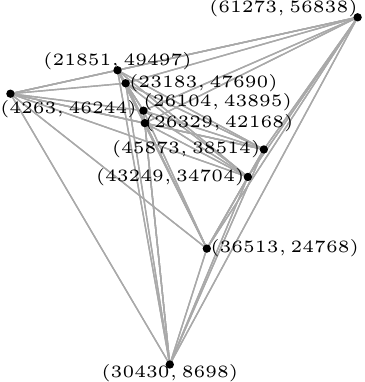}
  \caption{One universal point set for each $n=5,\dots,10$. Each pair of points
    is connected with a line segment.}
  \label{fig:universal_5_10}
\end{figure}

\newpage\bibliographystyle{plain}
\bibliography{universal}

\end{document}